\documentclass[11pt,a4paper]{article} 

\usepackage[a4paper,left=2cm, right=2cm, top=2.5cm, bottom=2.5cm]{geometry}
\usepackage[T1]{fontenc}
\usepackage{graphicx}
\usepackage{xcolor}
\definecolor{myblue}{RGB}{25,25,112}

\usepackage[english]{babel} 
\usepackage{natbib}
\usepackage[
    colorlinks,
    linkcolor=myblue!80!blue,
    citecolor=myblue!80!blue,
    urlcolor=myblue!60!blue
    ]{hyperref}
 
\usepackage{authblk}
\usepackage{amsmath,amssymb, amsthm}
\newtheorem{definition}{Definition}
\newtheorem{theorem}{Theorem}

\usepackage{bm}

\newcommand{\h}[1]{^{(#1)}}
\newcommand{\norm}[1]{{\left\vert\kern-0.25ex\left\vert #1      \right\vert\kern-0.25ex\right\vert_2}}
\newcommand{\normB}[1]{{\left\vert\kern-0.25ex\left\vert #1      \right\vert\kern-0.25ex\right\vert_B}}
\newcommand{\normsq}[1]{{\left\vert\kern-0.25ex\left\vert #1      \right\vert\kern-0.25ex\right\vert^2_2}} 
\newcommand{\mynorm}[1]{{\left\vert\kern-0.25ex\left\vert\kern-0.25ex\left\vert #1      \right\vert\kern-0.25ex\right\vert\kern-0.25ex\right\vert}} 
\newcommand{\scal}[2]{{\langle #1, #2\rangle_2}}
\newcommand{\scalB}[2]{{\langle #1, #2\rangle_B}}
\newcommand{\myscal}[2]{\langle \kern-0.25ex \langle #1, #2\rangle \kern-0.25ex \rangle}

\title{A General Framework for Multivariate Functional Principal Component Analysis of Amplitude and Phase Variation}

\title{A General Framework for Multivariate Functional Principal Component Analysis of Amplitude and Phase Variation}

\author[1]{Clara Happ}
\author[1]{Fabian Scheipl}
\author[2]{Alice-Agnes Gabriel}
\author[1]{Sonja Greven}
\affil[1]{Department of Statistics, LMU Munich, Munich, Germany}
\affil[2]{Department of Geophysics, LMU Munich, Munich, Germany}

\date{}

\begin{document}

\maketitle

\abstract{Functional data typically contains amplitude and phase variation. In many data situations, phase variation is treated as a nuisance effect and is removed during preprocessing, although it may contain valuable information.
In this note, we focus on joint principal component analysis (PCA) of amplitude and phase variation. 

As the space of warping functions has a complex geometric structure, one key element of the analysis is transforming the warping functions to $L^2(\mathcal{T})$. We present different transformation approaches and show how they fit into a general class of transformations. This allows us to compare their strengths and limitations. In the context of PCA, our results offer arguments in favour of the centered log-ratio transformation.

We further embed existing approaches from \citet{HadjipantelisEtAl:2015} and \citet{LeeJung:2017} for joint PCA of amplitude and phase variation into the framework of multivariate functional PCA, where we study the properties of the estimators based on an appropriate metric. The approach is illustrated through an application from seismology. }


\section{Introduction}

Functional data analysis \citep{RamsaySilverman:2005} is concerned with the broad field of data that comes in the form of functions $x_1 , \ldots , x_N$. For simplicity, consider here $x_i \in L^2(\mathcal{T}),~ i = 1 , \ldots , N$ with $\mathcal{T} = [a,b] \subset {\mathbb R},~ \eta = b - a > 0$.
 Warping approaches  \citep[see e.g.~][and references therein]{MarronEtAl:2015}  aim at decomposing the observed functions $x_i$ into  warping functions $\gamma_i$, that account for the phase variation in $x_i$, and registered functions $w_i$ that account for the amplitude variation in the data, via concatenation: $x_i(t) = (w_i \circ \gamma_i )(t) = w_i(\gamma_i(t))$. As a result, the main features of $w_i$ such as maximum and minimum points will typically be aligned, which reduces the variation and makes it easier to find common structures. The alignment is accomplished by the warping functions $\gamma_i$ that map the individual time of each observed curve $x_i$ to the global, absolute time of the registered curves.
Throughout this paper, we will consider the warping, i.e.\ the decomposition of $x_i$ into $w_i$ and $\gamma_i$, as given. In practical analyses, one would therefore have to choose an appropriate warping algorithm for the data before using our results. For a discussion on possible warping approaches, see e.g.\ \citet{MarronEtAl:2015}.

In many data situations, the main interest lies in the amplitude variation, meaning that warping is considered as part of the preprocessing and the subsequent analysis is  based on the registered functions $w_i$ only.
In others, phase variation is an integral part of the data that carries important information. It is thus worth to be incorporated into the analysis to gain deeper insights into the mechanisms generating the data. As an example, in our seismological application, the seismic waves do not only arrive with different intensity, but also with different time delays depending on their type and the geographical relation between the hypocenter and each seismometer. 
Moreover, in a recent result it was shown that warping may be affected by severe non-identifiability  in the sense that there may exist different representations $x_i = w_i \circ \gamma_i = \tilde w_i \circ \tilde \gamma_i$ with $w_i \neq \tilde w_i$ and $\gamma_i \neq \tilde \gamma_i$ \citep[to see this, set e.g.\ $\tilde w_i = x_i$ and $\tilde \gamma_i = id$, i.e.\ no warping]{ChakrabortyPanaretos:2017}. In order to keep the original information in $x_i$, it seems crucial to analyze both the warping and the registered function.

In this context, two methods that account for both amplitude and phase variation have been proposed for PCA, which is often used in functional data analysis for data exploration, dimension reduction and as a building block for methods such as regression: 
\citet{HadjipantelisEtAl:2015} transform the warping functions to the space of square integrable functions, $L^2(\mathcal{T})$, by differentiating them and taking the log. Next, they calculate a separate functional PCA for the transformed warping functions and the registered functions. 
Dependencies between the warping functions and the registered functions are then modeled via a linear mixed effects model of the principal component scores.
\citet{LeeJung:2017} also propose to transform the warping functions to $L^2(\mathcal{T})$, but then study a combined principal component analysis by ``glueing'' the discretized functions together.  They describe an optimal weighting scheme for the transformed warping functions based on dimension reduction and optimal reconstruction.
The main reason for transforming the warping functions in both approaches is that the space of warping functions on $\mathcal{T}$, $\Gamma(\mathcal{T})$, has a complex geometric structure, as it is for example not closed under addition or scalar multiplication and is not equipped with a natural scalar product, which is essential to define principal components. On the other hand, principal component analysis in $L^2(\mathcal{T})$ is well studied and can easily be applied to the transformed warping functions.

In this note, we discuss the role of the transformation of the warping functions to $L^2(\mathcal{T})$. We define a general class of transformations that includes existing  approaches from the literature as well as alternative transformations based on density functions on $\mathcal{T}$ and allows to compare their strengths and limitations. We offer theoretical arguments in favour of  the centered log-ratio (clr) transformation based on the Bayes space of densities \citep[see e.g.][]{EgozcueEtAl:2006, HronEtAl:2016} and illustrate this by means of a simulated toy example. Second, we embed the methods of \citet{HadjipantelisEtAl:2015} and \citet{LeeJung:2017} for the joint analysis of amplitude and phase variation into  a multivariate functional principal component framework which is based on univariate functional principal component analysis and takes correlations of warping functions and registered functions into account. For general transformations of the warping functions to $L^2(\mathcal{T})$, we study properties of resulting joint principal components based on an appropriate metric.
The proposed method based on the clr transformation is applied to data from a seismological computer experiment yielding spatially referenced high resolution time series of ground velocity measurements. In this context, the joint analysis of amplitude and phase variation is of particular interest as both contain relevant information on the propagation of seismic waves. The new method is shown to give promising results with a meaningful interpretation.

The manuscript proceeds as follows: In Section~\ref{sec:trafoFuns}, we review several transformations from $\Gamma(\mathcal{T})$ to $L^2(\mathcal{T})$ proposed in the literature and give reasons why the centered log-ratio approach should be preferred to other transformations when defining PCA for warping functions.
Section~\ref{sec:Frechet} embeds the existing  methods for joint analysis of amplitude and phase variation into the framework of multivariate functional principal component analysis \citep[MFPCA]{HappGreven:2018} and gives new insights into the properties of the joint principal components. 
The theoretical results are illustrated in Section~\ref{sec:earthquake} by means of data from a  from a multi-physics computational seismology experiment.

\section{Transformation approaches for warping functions}
\label{sec:trafoFuns}

The space of warping functions 
\[
\Gamma(\mathcal{T}) = \{ \gamma: \mathcal{T} \to \mathcal{T} \colon \gamma \text{ is a diffeomorphism\footnotemark}, 
\, \gamma(a) = a,\, \gamma(b) = b\}
\]
has a complex, non-Euclidean geometric structure \citep{LeeJung:2017, SrivastavaBook}, as discussed before. 
In order to simplify the geometry of $\Gamma(\mathcal{T})$ and to derive principal components, it is natural to consider mappings $\Psi: \Gamma(\mathcal{T}) \to L^2(\mathcal{T})$, which allow to transform warping functions to square-integrable functions. In a second step, one can calculate principal components in the well studied space $L^2(\mathcal{T})$ and transform the results back by the inverse map $\Psi ^{-1}: L^2(\mathcal{T}) \to \Gamma(\mathcal{T})$. 
In this section, we present different approaches for this transformation and discuss their strengths and limitations within a new, general framework. Note that here we focus on PCA for the warping functions only and will return to the joint approach in the next section.
\textbf{\footnotetext{A \textit{diffeomorphism} describes a smooth and strictly increasing function \citep{MarronEtAl:2015}. In other words, the warping function $\gamma$ is required to be invertible and should not allow for abrupt jumps or for ``travelling back in time''.}}

\textbf{Square-root velocity transformation:}
For the mapping $\Psi$, \citet{LeeJung:2017} propose the square-root velocity  framework \citep{Tucker:2014,SrivastavaBook}, which consists of two steps: First, the warping functions are mapped to the positive orthant of the (scaled) unit sphere in $L^2(\mathcal{T})$, $S_+^\infty (\mathcal{T}) = \{s \in L^2(\mathcal{T}) \colon \normsq{s} = \eta, s \geq 0\}$. This is realized by the square-root velocity function $\operatorname{SRVF}: \gamma \mapsto \sqrt{\gamma'}$, where $\gamma'$ denotes the first derivative of $\gamma$. It can be shown that this is a bijection from $\Gamma(\mathcal{T})$ to $S_+^\infty (\mathcal{T})$. In order to map the transformed warping functions $q =  \sqrt{\gamma'}$ to $L^2(\mathcal{T})$, choose some $\mu \in S^\infty_+(\mathcal{T})$ and approximate $q$ by functions in the tangent space associated with $\mu$, $T_\mu(\mathcal{T}) = \{v \in L^2(\mathcal{T}): \scal{v}{\mu} = 0\}$. This second transformation step is done via the mapping 
\begin{equation}
\tilde \psi_{S,\mu}: S^\infty_+(\mathcal{T})  \to T_\mu(\mathcal{T}), \qquad
q  \mapsto \frac{\theta}{\eta^{1/2}\sin(\theta)} (q - \cos(\theta) \mu) ,\qquad \theta = \cos ^{-1}\left(\frac{\scal{q}{\mu}}{\eta}\right).
\label{eq:psiS}
\end{equation}
A natural choice for $\mu$ is the Karcher or Fr{\'{e}}chet mean of $q$ \citep{Tucker:2014}. For the special case of warping functions, one may also choose the constant function $q_0(t) \equiv 1$, which is associated with  identity warping $\gamma_0(t) = t$ \citep{LeeJung:2017}. The back transformation to $\Gamma(\mathcal{T})$ is again in two steps: First, apply 
\begin{equation}
\tilde \psi_{S,\mu}^{-1}:  T_\mu(\mathcal{T}) \to S^\infty(\mathcal{T}) , \qquad
v  \mapsto \cos(\norm{v})\mu + \eta^{1/2} \sin(\norm{v}) \frac{v}{\norm{v}}
\label{eq:psiSinv}
\end{equation}
to map $v$ to the sphere $S^\infty(\mathcal{T}) =\{s \in L^2(\mathcal{T}) \colon \normsq{s} = \eta \}$. Then, the results are mapped back to the space of warping functions via $\operatorname{SRVF}^{-1} \colon  S^\infty(\mathcal{T}) \to \Gamma(\mathcal{T})$ with $\operatorname{SRVF}^{-1}(s)(t) = a + \int_0^t s(u)^2 \mathrm{d} u$. 
The overall mapping from $\Gamma(\mathcal{T})$ to $L^2(\mathcal{T})$ is hence given by  $\Psi = \tilde \psi_{S,\mu} \circ \operatorname{SRVF}$ with the inverse $\Psi^{-1} ={\operatorname{SRVF}^{-1}} \circ {\tilde \psi_{S,\mu}^{-1}}$.

The square-root velocity framework has been shown to be advantageous in statistical shape analysis, particularly when combined with the Fisher-Rao metric, which is invariant under warping \citep{SrivastavaBook}. It can also be used to define warping-invariant PCA \citep{Tucker:2014}. 
For the calculation of principal components for warping functions as described above, however, it has two serious shortcomings. 
First, $\tilde \psi_{S,\mu}$ and $\tilde \psi_{S,\mu}^{-1}$ are undefined for the rather interesting points $\mu \in S_+^\infty(\mathcal{T})$ and $v_0 \equiv 0 \in T_\mu(\mathcal{T})$. While, theoretically, they can easily be completed by applying l'H{\^{o}}pital's rule, computational instabilities for functions close to $\mu$ and $v_0$ often occur in practice. At the same time, the projection from the positive orthant of the sphere $S_+^\infty(\mathcal{T})$ to the tangent space $T_\mu(\mathcal{T})$ is a local approximation that works best in the vicinity of $\mu$. On the one hand, this shows that the choice of $\mu$ is important and on the other hand requires the data to be neither too close nor too far from $\mu$.
Second, and even more importantly, the mappings $\tilde \psi_{S,\mu}$ and $\tilde \psi_{S,\mu}^{-1}$ are not inverse to each other as
\begin{equation}
\operatorname{dom}(\tilde \psi_{S,\mu} ) = S_+^\infty(\mathcal{T}) \subsetneq S^\infty(\mathcal{T})  = \operatorname{im}(\tilde \psi_{S,\mu}^{-1}) 
\qquad \text{and} \qquad
\operatorname{im}(\tilde \psi_{S,\mu} )  \subsetneq T_\mu(\mathcal{T})  = \operatorname{dom}(\tilde \psi_{S,\mu}^{-1}).
\label{eq:imdomPsiS} 
\end{equation}
While a theoretical proof is given in the appendix, the main problem can be seen at a glance in Fig.~\ref{fig:imPsiS}:
\begin{figure}
\begin{center}
\includegraphics[width = 0.35\textwidth]{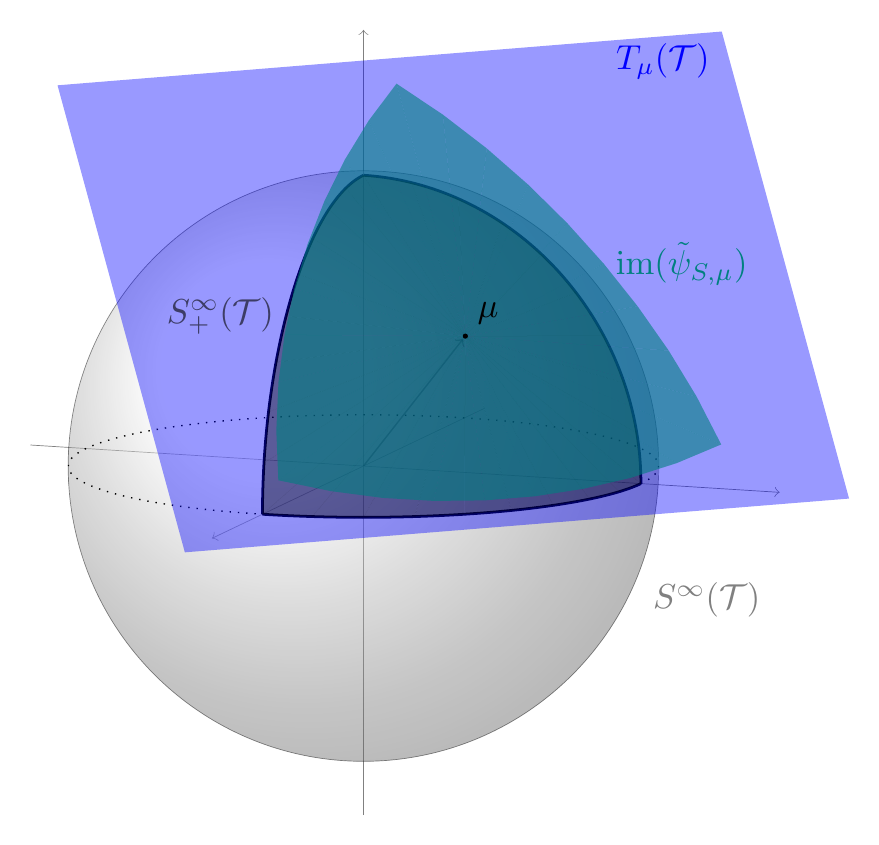}
\end{center}

\caption{
Illustration of the sphere $S^\infty$ and its positive orthant $S_+^\infty$ (dark grey) and the projection of this orthant to the tangent space $T_\mu$ via $\tilde \psi_{S,\mu}$ (green).
}
\label{fig:imPsiS}
\end{figure}
The domain $S_+^\infty(\mathcal{T})$ and the image $\operatorname{im}(\tilde \psi_{S,\mu})$ of the projection $\tilde \psi_{S,\mu}$ are proper subsets, and therefore ``smaller'', than the full sphere $S^\infty(\mathcal{T})$ and the tangent space $T_\mu(\mathcal{T})$, which form the image and domain of the back transformation $\tilde \psi_{S,\mu}^{-1}$.
For all SRVFs $q \in S_+^\infty(\mathcal{T})$ which have been projected to $T_\mu(\mathcal{T})$ by $\tilde \psi_{S,\mu}$, there is a unique mapping back to $S_+^\infty(\mathcal{T})$, since the restriction $\tilde \psi_{S,\mu}^{-1} \vert_{\operatorname{im}(\tilde \psi_{S,\mu} )}$ is a bijection \citep{SrivastavaBook}. 
 As soon as one leaves the subset $\operatorname{im}(\tilde \psi_{S,\mu} )$, however, there is no guarantee that the image of an arbitrary function $v \in T_\mu$ under $\tilde \psi_{S,\mu}^{-1}$ is again in the positive orthant of $S^\infty(\mathcal{T})$ and thus represents a valid SRVF. As $\operatorname{SRVF}^{-1}$ is not injective, the final transformation back to $\Gamma(\mathcal{T})$ will always yield a valid warping function, but potentially with atypical structure.
In order to see that leaving $\operatorname{im}(\tilde \psi_{S,\mu} )$ is likely to happen in practice, consider the following toy example:

\textit{Example:}
Let $\mathcal{T} = [0,1]$. Then $\gamma(t) = t^k$ clearly is a valid warping function for all $k > 0$. The associated square-root velocity function is given by $q(t) = \sqrt{\gamma'(t)} = \sqrt{k t^{k-1}}$. Choosing the SRVF of the identity warping function $\gamma_0$ for $\mu$, the approximation of $q$ in the tangent space $T_\mu$ is given by
$
v = \tilde \psi_{S,\mu}(q)$ using $\theta = \cos^{-1}(\scal{q}{\mu}) = \cos ^{-1}(\frac{2\sqrt{k}}{k + 1})$. Fig.~\ref{fig:exampleSRVF} illustrates $N = 50$ warping functions $\gamma_{i}(t) = t^{k_i},~i = 1, \ldots , N$ together with their SRVFs $q_i$ and the tangent-space approximations $v_i$ (grey curves). The values $k_i$ are generated according to a $\operatorname{Ga}(5,5)$ distribution, such that ${\mathbb E}(k_i) = 1$, which corresponds to identity warping, and $\operatorname{Var}(k_i) = 0.2$ for intermediate variation.
The first eigenfunction $\hat \phi_1$ of the tangent-space approximations in $L^2(\mathcal{T})$ explains more than $95\%$ of the variability in $v_i$ and is shown in black in Fig.~\ref{fig:exampleSRVF}. Note that $\hat \phi_1$ takes values below $-1$ in the left part of $\mathcal{T}$, meaning that $\hat \phi_1$ is outside $\operatorname{im}(\tilde \psi_{S,\mu} )$. Since the eigenvalue $\hat \lambda_1 = 0.036$ is rather small, visualizing the effect of the first principal component as scaled perturbation from the mean ($\hat \mu \pm \hat \lambda_1^{1/2} \hat \phi_1$, blue curves in Fig.~\ref{fig:exampleSRVF}) yields valid SRVFs and sensible warping functions after back transformation. However, transforming the original principal component $\hat \phi_1$ back via $\tilde \psi_{S,\mu}^{-1}$ yields a function $q_\phi \in S^\infty(\mathcal{T})$, but outside $S_+^\infty(\mathcal{T})$, as it has negative values and thus does not represent typical structures of the $q_i$. The back transformation to $\Gamma(\mathcal{T})$ yields a valid warping function $\gamma_\phi$, but again with a very different form than the original data. The green curve in the right of  Fig.~\ref{fig:exampleSRVF} has a more pronounced curvature than the curves in the original sample. Reconstructing the green curve based on the first principal component yields a function that is shifted towards the identity warping function and does not represent the original green curve well.

This simple example shows that whenever one uses the projected SRVFs $v_i$ for statistical analyses in $T_\mu(\mathcal{T})$ whose results are not guaranteed to stay within $\operatorname{im}(\tilde \psi_{S,\mu} )$, there is a risk of obtaining atypical results on the level of the SRVFs and of the warping functions. Since $\operatorname{im}(\tilde \psi_{S,\mu} )$ is not closed under vector space operations in $T_\mu(\mathcal{T})$, this risk occurs for example when calculating principal components of $v_i$, representing $v_i$ in a truncated basis expansion or when using $v_i$ as covariates in a regression. In all these cases, our results show that it is very likely that transforming the results of the analysis from $T_\mu(\mathcal{T})$ to $\Gamma(\mathcal{T})$ yields warping functions with atypical structure, which may be hard to interpret. As $T_\mu(\mathcal{T})$ is a local approximation, anomalies  may even occur if the results are within $\operatorname{im}(\tilde \psi_{S,\mu} )$, but close to the boundary. An example for this can be seen in the prediction of the green curve in Fig.~\ref{fig:exampleSRVF}.

\begin{figure}
\begin{center}
\includegraphics[width = \textwidth]{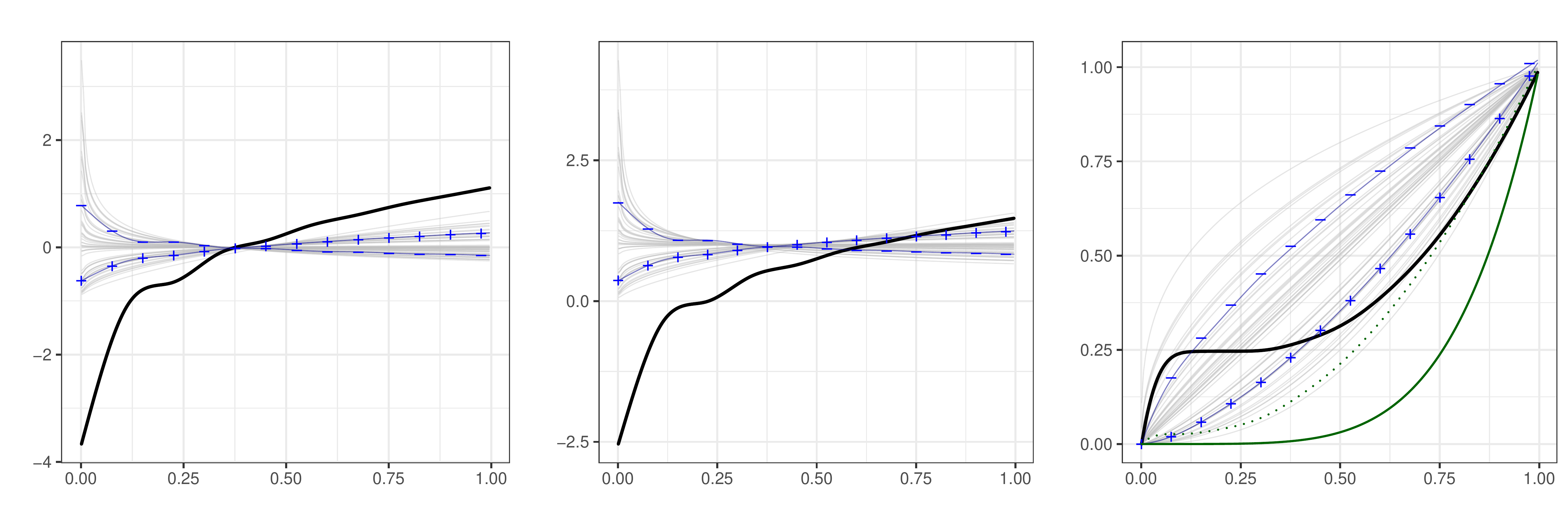}
\end{center}

\caption{
PCA of warping functions based on $\tilde \psi_{S,\mu}$ for simulated warping functions and different stages of transformation. The grey functions represent the (transformed) data, the black function the first principal component $\hat \phi_1$. Blue lines visualize the effect of the first principal component as perturbation from the mean ($\hat \mu \pm \hat \lambda_1^{1/2} \hat \phi_1$).
Left: Tangent space $T_\mu(\mathcal{T})$. Center: Sphere $S^\infty(\mathcal{T})$. Right: Space of warping functions $\Gamma(\mathcal{T})$.
The green curves corresponds to a new observation (solid) and its reconstruction (dotted) based on the first principal component.
}
\label{fig:exampleSRVF}
\end{figure}

\textbf{Bayes space transformation:}
Due to their typical characteristics, warping functions $\gamma$ can also be interpreted as generalized cumulative distribution functions of continuous random variables $X: \Omega \to \mathcal{T}$ in the sense that $\gamma(a) = a, \gamma(b) = b$ and $\gamma$ is monotonically increasing. Taking the first derivative $\gamma'$ yields a unique (scaled) probability density function on $\mathcal{T}$. While this has been mentioned in \citet{HadjipantelisEtAl:2015}, they do not exploit the geometric structure of the transformed functions. As shown in \citet{EgozcueEtAl:2006}, density functions, or, more precisely, equivalence classes of these functions, form a vector space, known as Bayes Hilbert space $B^2(\mathcal{T})$.  We propose to make use of this specific space to define a PCA for warping functions.

Following \citet{EgozcueEtAl:2006}, two functions $f,g \in B^2(\mathcal{T})$ are equivalent in the Bayes Hilbert space, if they are proportional ($f = \alpha g$). A natural representative for each class is given by the function integrating to $\eta = b-a$, that we interpret as the derivative of a warping function. For $f,g \in B^2(\mathcal{T})$ and $\alpha \in {\mathbb R}$, operations on $B^2(\mathcal{T})$ are defined as \citep{EgozcueEtAl:2006}
\begin{align*}
(f \oplus g)(t) &= \frac{f(t)g(t)}{\int_a^b f(s) g(s) \mathrm{d} s}
\qquad
(a \odot f)(t) = \frac{f(t)^\alpha}{\int_a^b f(s)^\alpha \mathrm{d} s} \\
\scalB{f}{g} &= \frac{1}{2\eta} \int_a^b \int_a^b \log\left(\frac{f(x)}{f(y)}\right)\log\left(\frac{g(x)}{g(y)}\right) \mathrm{d} y \mathrm{d} x.
\end{align*}
The centred log-ratio transform $\psi_B: B^2(\mathcal{T}) \to L^2(\mathcal{T})$ 
\[
\psi_B(f)(t) = \log(f(t)) - \frac{1}{\eta} \int_a^b \log(f(x)) \mathrm{d} x
\qquad
\psi_B^{-1}(f)(t) =\eta \cdot \frac{\exp(f(t))}{\int_a^b \exp(f(s)) \mathrm{d} s}
\]
is an isometric isomorphism with respect to the norm induced by $\scalB{\cdot}{\cdot}$.
Therefore $\psi_B(f \oplus g) = \psi_B(f) + \psi_B(g)$, $\psi_B(\alpha \odot f) = \alpha \cdot \psi_B(f)$ and $\scalB{f}{g} = \scal{\psi_B(f)}{\psi_B(g)}$ with the standard operations ($+, \cdot, \scal{\cdot}{\cdot}$) on $L^2(\mathcal{T})$ \citep{EgozcueEtAl:2006}.  One may note that \citet{HadjipantelisEtAl:2015} use a similar log-transformation for warping functions as $\psi_B$, but without the integral term. 

The geometric structure of $B^2(\mathcal{T})$ allows to easily transfer PCA to density functions by transforming them to $L^2(\mathcal{T})$, calculating the PCA in this space and transforming the result back to $B^2(\mathcal{T})$ \citep{HronEtAl:2016}. Since the mapping between a warping function $\gamma$ and the associated density is one-to-one, we may directly apply PCA based on this Hilbert space transformation to warping functions, too. 
The resulting principal components for warping functions give more meaningful results in our toy example, as shown in Fig.~\ref{fig:exampleBayes}. Moreover, the reconstruction of the green curve based on the first principal component is almost perfect.

As shown in \citet{HronEtAl:2016}, the transformation $\psi_B$ does also not map into full $L^2(\mathcal{T})$, but into the subspace $U_B(\mathcal{T}) = \{v \in L^2(\mathcal{T}) \colon \int_a^b v(s) \mathrm{d} s = 0\}$. This space is equivalent to the space $T_\mu(\mathcal{T})$ in the SRVF approach, if we choose $\mu$ the SRVF associated with the identity warping function, as proposed in \citet{LeeJung:2017}.
An important difference between the SRVF approach and the Bayes space transformation, however, is that the centered log-ratio transformation $\psi_B$ gives a one-to-one transformation from $B^2(\mathcal{T})$ to $U_B(\mathcal{T})$, meaning that $\operatorname{im}(\psi_B)$ covers full $U_B(\mathcal{T})$, which is closed under vector space operations. Leaving $\operatorname{im}(\psi_B)$ in principal component analysis or regression is thus not possible. Particularly, this means that the constraint $\int_a^b v(s) \mathrm{d} s = 0$ is automatically fulfilled for the principal components and does not need to be explicitly imposed in the implementation. Moreover, as an isometric isomorphism, $\psi_B$ preserves a lot more of the geometric structure of the transformed warping functions than $\tilde \psi_{S,\mu}$, which only preserves distances.

\begin{figure}
\begin{center}
\includegraphics[width = \textwidth]{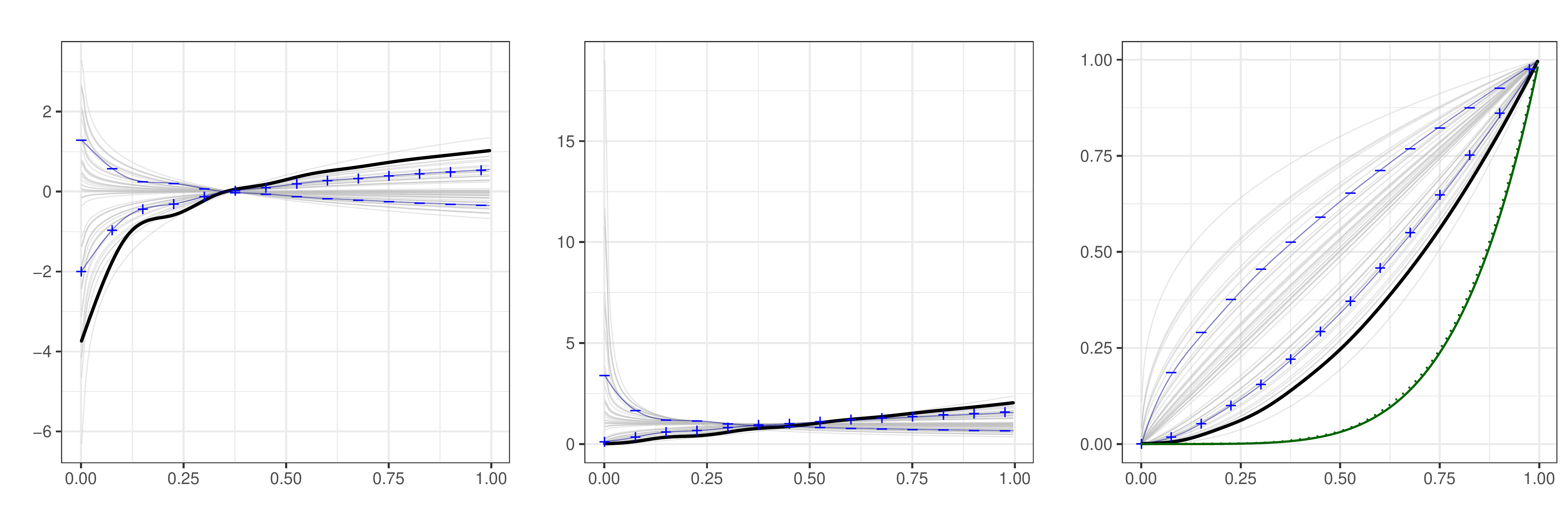}
\end{center}
\caption{
PCA of warping functions based on $\psi_B$ for the same warping functions as in Fig.~\ref{fig:exampleSRVF} and different stages of transformation. The grey functions represent the (transformed) data, the black function the first principal component $\hat \phi_1$. Blue lines visualize the effect of the first principal component as perturbation from the mean ($\hat \mu \pm \hat \lambda_1^{1/2} \hat \phi_1$).
Left: $L^2(\mathcal{T})$. Center: Bayes space $B^2(\mathcal{T})$. Right: Space of warping functions $\Gamma(\mathcal{T})$.
The green curves corresponds to a new observation (solid) and its reconstruction (dotted) based on the first principal component.
}
\label{fig:exampleBayes}
\end{figure}

\textbf{Other transformations:}
Principal component analysis  of density functions based on transformations to $L^2(\mathcal{T})$ has been considered more generally in \citet{PetersenMueller:2016}. They present two alternative transformations that we can again apply to the derivatives of warping functions.  
The log-hazard transformation is
\[
\psi_H(f) = \log \left( \frac{f}{1 - F} \right)
\]
with $F$ the warping function scaled to $[0,1]$ and $f$ the associated density. Since hazard functions are known to diverge  at the right endpoint of $\mathcal{T}$, the transformation is considered only on a subinterval $[a, b - \delta \eta]$ with $\delta$ a threshold parameter. For the back transformation, uniform weight is assigned to $t \in ( b - \delta \eta, b] $.
Alternatively, \citet{PetersenMueller:2016} consider the log quantile density transformation
\[
\psi_Q(f) = - \log(f(Q))
\]
with $Q$ the inverse (quantile) function associated with $F$. Of the two, \citet{PetersenMueller:2016}  recommend to use the log-quantile transformation. 

Regarding the practical usability in the context of warping functions, the log-hazard transformation $\psi_H$ may be highly influenced by the threshold parameter $\delta$, as shown in Fig.~\ref{fig:exampleLH} in the appendix  for the curves of our toy example. There we choose $\delta = 0.05$, which corresponds to cutting the warping functions at the $95\%$ quantile of the identity warping function. On the other hand, the log-quantile density transformation $\psi_Q$ requires numerical inversion of the warping functions $\gamma$, which may also lead to instabilities. Regarding the reconstruction of the new green curve in Fig.~\ref{fig:exampleLH}, the principal component analysis based on log-hazard transformation performs similarly well as the one based on the clr transformation. By contrast, for the log-quantile transformation, the prediction is rather poor, as it is shifted considerably towards the identity warping. This, however, is in contradiction with the recommendation given in \citet{PetersenMueller:2016}, as it would favour the log-hazard transformation. 

\textbf{General framework:}
All discussed transformations can be subsumed as transformations $\Psi_\ast: \Gamma(\mathcal{T}) \to L^2(\mathcal{T})$ with $\Psi_\ast = \psi_\ast \circ D$. Here $D$ denotes the differential operator that maps a warping function $\gamma$ to its density $\gamma'$. The restriction $\gamma(a) = a$ and $\gamma(b) = b$ makes this a one-to-one mapping. The mapping $\psi_\ast$ is a transformation from the space of density functions to $L^2(\mathcal{T})$ and  depends on the transformation. In particular, the SRVF approach can also be considered as such a transformation with $\psi_{S,\mu}(q) = \tilde \psi_{S,\mu}(\sqrt{q})$ and $\psi_{S,\mu}^{-1} (v) = [\tilde \psi_{S,\mu}^{-1}(v)]^2$. For the Bayes space approach, $\psi_\ast$ corresponds to the centred log-ratio transform $\psi_B$ and for the log-hazard transformation and the log-quantile density transformation proposed in \citet{PetersenMueller:2016}, $\psi_\ast$ corresponds to $\psi_H$ and $\psi_Q$, respectively. In order to compare the different transformations, one can therefore focus on $\psi_\ast$ only. 

As discussed before, the SRVF transformation via $\tilde \psi_{S,\mu}$ does not necessarily preserve structure in the transformed warping functions, as back  transformation via  $\tilde \psi_{S,\mu}^{-1}$ may result in functions outside $S_+^\infty(\mathcal{T})$, having negative values. The alternative formulation $\Psi_S = \psi_{S,\mu} \circ D$ respects this structure better, as $\psi_{S,\mu}^{-1}$ always yields valid density functions. However, this reformulation does not change the fact that $\psi_{S,\mu}$ is not surjective and therefore PCA based on the SRVF transformation does not necessarily return structures that resemble those of the original data. Moreover, the transformation may not be computationally robust for functions close to $\mu$ (in $S_+^\infty(\mathcal{T})$) or $v_0$ (in $L^2(\mathcal{T})$).
The centred-log-ratio transform $\psi_B$ is an isometric isomorphism between $B^2(\mathcal{T})$ and $L^2(\mathcal{T})$. It hence preserves the Hilbert space structure and seems particularly suitable for PCA. Due to the log-transformation, numerical instabilities may occur in regions in which $\gamma'(t)$ is close to zero,  where the warping function $\gamma$ is near the degenerate constant case.
The alternative transformations proposed in \citet{PetersenMueller:2016} do not represent isometric isomorphisms, but preserve the space structures by construction. However, they may be affected by numerical instabilities, as shown in the toy example (Fig.~\ref{fig:exampleLH}).
Overall, the clr transformation seems to be the most suitable transformation for principal component analysis.

\section{Modes of joint variation in amplitude and phase}
\label{sec:Frechet}
In order to study the joint variation of amplitude and phase in the data, \citet{LeeJung:2017} propose to concatenate the registered functions $w_i$ and the transformed warping functions $v_i = \Psi(\gamma_i)$, including a weighting parameter $C$:
\[
g^C_i(t) = 
\begin{cases}
w_i(t) \quad t \in [a,b) \\
C v_i(t-(b-a)) \quad t \in [b, 2b -a).
\end{cases}
\]
Remember that the warping is assumed to be given, i.e. in practical applications, $w_i$ and $\gamma_i$ have to be replaced by estimates obtained from the data.

For the transformation, \citet{LeeJung:2017} use $\Psi = \Psi_S$, but the method can be directly transferred to general transformations $\Psi \colon \Gamma(\mathcal{T}) \to L^2(\mathcal{T})$. The principal component analysis is then based on these concatenated functions, $g^C_i$, evaluated  on a fine grid. The resulting PCs are then separated and transformed back for the final interpretation. 

As noted in \citet{LeeJung:2017}, one may  alternatively use methods for multivariate functional principal component analysis \citep[MFPCA, ][]{ChiouEtAl:2014, HappGreven:2018}. These approaches are indeed more appropriate as they better reflect the characteristic nature of the data in terms of bivariate functions ${\bm z}_i = (w_i, v_i) \in L^2(\mathcal{T}) \times L^2(\mathcal{T}) =: \mathcal{H}$ and therefore we will use them in the following. Moreover, the bivariate formulation  ensures that we do not need to deal with potential incontinuities of $g^C_i$ at $b$. The approach in \citet{HappGreven:2018} estimates the multivariate functional PCA via a  univariate functional principal component analysis followed by a PCA of the combined score vectors based on a theoretical equivalence result. In this sense, it is similar to the approach in \citet{HadjipantelisEtAl:2015}, who also calculate a separate functional principal component analysis, but then use a linear mixed model to capture joint variation in the scores. The two-step approach in \citet{HappGreven:2018} can also easily be applied to warped curves in ${\mathbb R}^2$ or ${\mathbb R}^3$ or warped images. In the case of two-dimensional curves one would for example have a registered function for the $x$- and $y$-coordinate and a warping function. One would then apply the method to the trivariate vector of all three functions.

The weighting factor $C$ can directly be included into a weighted scalar product $\myscal{\bm{z}_i}{\bm{z}_j}_w =  \scal{w_i}{w_j} + C^2 \scal{v_i}{v_j}$. This corresponds to the standard scalar product on $\mathcal{H}$ for the bivariate function $\tilde { \bm z}_i = (w_i, C v_i)$, which is mimicked by $g^C_i$. 
In analogy to \citet{LeeJung:2017}, an optimal choice for $C$ can be found by optimizing the reconstruction based on the first $M$ principal component. More specifically, let the truncated and estimated Karhunen-Lo{\`{e}}ve representation of the bivariate function ${\bm z}_i $ be
\[
\bm{  \hat  z}_i^{[M]} =  { \bm{ \hat \mu_z}} + \sum_{m = 1}^M \hat \rho_{im}   \hat {\bm \phi}_m
\]
with  $\bm{ \hat \mu_z} = (\bar w, \bar v)$ the bivariate mean function, $\hat  {\bm \phi}_m = \left(\hat \phi_m\h{w}, \hat \phi_m\h{v}\right)$ the $m$-th bivariate principal component and $\hat \rho_{im} = \myscal{{\bm z}_i - \bm{ \hat \mu_z}}{ \hat{\bm \phi}_m}_w$ the observation-specific principal component scores \citep[see][for further details]{HappGreven:2018}. Denoting the elements of $\bm{  \hat  z}_i^{[M]}$ by $\hat w_i^{[M]}$ and $\hat v_i^{[M]}$, respectively, a reconstruction of $x_i = w_i \circ \gamma_i $ is given by $ \hat x_i^{[M]} = \hat w_i^{[M]}  \circ \Psi ^{-1}\left(\hat v_i^{[M]}\right) $. For given $M$, an optimal value of $C$ can be found as
\[
\hat C = \operatorname{arg\,min}_{C > 0} \frac{1}{N} \sum_{i = 1}^N \normsq{x_i - \hat x_i^{[M]}}\,.
\]
This choice makes sure that the principal component analysis also yields optimal reconstructions with respect to the originally observed data.

The eigenvalue $\nu_m$ associated with the principal component $ {\bm \phi}_m$ can be interpreted as the amount of variability in  ${\bm z}_i $ that is explained by this principal component. Typically, an optimal choice for $M$ is obtained by 
\begin{equation}
\hat M = \min \left\lbrace M \in {\mathbb N}:  \sum_{m = 1}^M \hat \nu_m \middle/ \sum_{m = 1}^\infty \hat \nu_m > \tau \right\rbrace, 
\label{eq:Mhat}
\end{equation}
where $\tau$ is a threshold for the minimum proportion of variability in $\bm{z}_1 , \ldots , \bm{z}_N$ that is to be explained by the first $\hat M$ components (typically $\tau = 0.95$ or $\tau = 0.99$). In practice, the infinite sum in the denominator is replaced by the sum over all calculated eigenvalues. 

As seen in expression \eqref{eq:Mhat}, the value $\hat M$ relates to the variability in $\bm{z}_i = (w_i, v_i)$ and not in the decomposition of the original data, $\bm{\zeta}_i = (w_i, \gamma_i)$. In the following, we show that the ratio in \eqref{eq:Mhat} can also be interpreted as a proportion of Fr{\'{e}}chet variances for $\bm{\zeta}_i $ by defining an appropriate metric. This is a new result and has not been considered in \citet{LeeJung:2017} or \citet{HadjipantelisEtAl:2015}. The Fr{\'{e}}chet mean and variance generalize the concept of mean and variance to random variables on general metric spaces and have been studied  in the context of random densities in $B^2(\mathcal{T})$ by \citet{vanDenBoogaartEtAl:2014, PetersenMueller:2016}, among others.

\begin{definition}
\label{def:metric}
For $\bm{\zeta}_1, \bm{\zeta}_2 \in \mathcal{G} := L^2(\mathcal{T}) \times \Gamma(\mathcal{T})$ with $\bm{\zeta}_i = (w_i, \gamma_i)$ and $\bm{z}_i = (w_i, \Psi(\gamma_i)), i = 1,2$ define 
\[
d(\bm{\zeta}_1, \bm{\zeta}_2) = \mynorm{\bm{z}_1-\bm{z}_2}_w = \left[\normsq{w_1 - w_2} + C^2\normsq{ \Psi(\gamma_1)- \Psi(\gamma_2)} \right]^{1/2} \qquad \text{with}~C > 0.
\]
\end{definition}

As $d$  is induced by the norm $\mynorm{\cdot}_w$ on $\mathcal{H}$, it is immediately clear that $d$ defines a valid metric on $\mathcal{G}$. Based on this metric, we obtain the following result:

\begin{theorem}
Consider a random element $\bm{\zeta} \in \mathcal{G}$ with $\bm{\zeta} = (w, \gamma)$. Then the Fr{\'{e}}chet mean ${\mathbb E}_d(\bm{\zeta})$ and Fr{\'{e}}chet variance $\operatorname{Var}_d(\bm{\zeta})$ of $\bm{\zeta}$ based on the metric $d$ are given by
\[
{\mathbb E}_d(\bm{\zeta}) = \left({\mathbb E}(w), \Psi^{-1}({\mathbb E}(\Psi(\gamma)))\right), 
\qquad
\operatorname{Var}_d(\bm{\zeta}) = \int_a^b \operatorname{Var}(w(t)) + C^2\operatorname{Var}(\Psi(\gamma(t))) \mathrm{d} t,
\]
where ${\mathbb E}(f)(t) = {\mathbb E}[f(t)]$ denotes the usual mean function for random functions $f \in L^2(\mathcal{T})$.
\end{theorem}

\begin{proof}
Fr{\'{e}}chet mean:
\begin{align*}
{\mathbb E}_d(\bm{\zeta})
& = \operatorname{arg\,inf}_{\bm{\theta} \in \mathcal{G}} {\mathbb E} \left[d(\bm{\zeta},\bm{\theta})^2 \right]
= \operatorname{arg\,inf}_{(w_\theta, \gamma_\theta) \in L^2(\mathcal{T}) \times \Gamma(\mathcal{T})} {\mathbb E} \left[
\normsq{w - w_\theta} + C^2 \normsq{\Psi(\gamma) - \Psi(\gamma_\theta)}  \right] \\
&= \left( \operatorname{arg\,inf}_{w_\theta \in  L^2(\mathcal{T})}  {\mathbb E} \left[ \normsq{w - w_\theta} \right],
 \operatorname{arg\,inf}_{\gamma_\theta \in \Gamma(\mathcal{T})} {\mathbb E} \left[
\normsq{\Psi(\gamma) - \Psi(\gamma_\theta)}\right] \right)
 = \left( {\mathbb E}(w), \Psi^{-1} ({\mathbb E}(\Psi(\gamma))) \right) ,
\end{align*}
where the last step can be seen e.g.\ by completing the squares.

Fr{\'{e}}chet variance: 
\begin{align*}
\operatorname{Var}_d(\bm{\zeta})  &= {\mathbb E}(d(\bm{\zeta}, {\mathbb E}_d(\bm{\zeta}))^2)
 = {\mathbb E} \left[\normsq{w - {\mathbb E}(w)} + C^2\normsq{\Psi(\gamma) - \Psi(\Psi^{-1}({\mathbb E}(\Psi(\gamma))))} \right] \\
& = {\mathbb E} \left[ \int_a^b (w(t) - {\mathbb E}(w(t)))^2   \mathrm{d} t  +  C^2\int_a^b \left(\Psi(\gamma(t)) - {\mathbb E}(\Psi(\gamma(t)))\right)^2 \mathrm{d} t  \right] \\
& =  \int_a^b \operatorname{Var}(w(t)) \mathrm{d} t  + C^2 \int_a^b \operatorname{Var}(\Psi(\gamma(t))) \mathrm{d} t
\end{align*}
\end{proof}

The Fr{\'{e}}chet mean of $\bm{\zeta} $ based on $d$ is hence given by the mean function of the transformed data $\bm{z} $, with the warping part back-transformed to $\Gamma(\mathcal{T})$ via $\Psi^{-1}$. Similarly, the Fr{\'{e}}chet variance can be obtained by a weighted sum of the integrated pointwise variances of the elements of $\bm{z}$.

Using the infinite bivariate Karhunen-Lo{\`{e}}ve expansion ${\bm z} = {\bm \mu_z} + \sum_{m = 1}^\infty \rho_{m}  {\bm \phi}_m$ of a random element ${\bm z} \in \mathcal{H}$, the Fr{\'{e}}chet variance of the associated $\bm{\zeta}$ can be rewritten as
\begin{align*}
\operatorname{Var}_d(\bm{\zeta}) & =  \int_a^b \operatorname{Var}(w(t)) \mathrm{d} t  + C^2 \int_a^b \operatorname{Var}(\Psi(\gamma(t))) \mathrm{d} t \\
&= \int_a^b \operatorname{Var}(\sum_{m = 1}^\infty \rho_m \phi_m\h{w}(t)) + C^2 \operatorname{Var}(\sum_{m = 1}^\infty \rho_m \phi_m\h{v}(t))) \mathrm{d} t \\
& = \sum_{m = 1}^\infty  \nu_m \int_a^b  [\phi_m\h{w}(t)^2 + C^2 \phi_m\h{v}(t)^2] \mathrm{d} t  =\sum_{m = 1}^\infty  \nu_m  \mynorm{{\bm \phi}_m}_w^2 = \sum_{m = 1}^\infty  \nu_m,
\end{align*}
since the scores $\rho_m$ are uncorrelated with variance $\operatorname{Var}(\rho_m) = \nu_m$ and the principal components have weighted norm $1$. Thus, the Fr{\'{e}}chet variance of $\bm{\zeta}$ equals the total variance in $\bm{z}$.

For given $M$, we obtain the fundamental variance decomposition
\[
\operatorname{Var}_d(\bm{\zeta}) = \sum_{m = 1}^M \nu_m + \sum_{m = M+1}^\infty \nu_m 
= \operatorname{Var}_d(\bm{\zeta}^{[M]}) + \operatorname{Var}_d(\bm{\zeta} - \bm{\zeta}^{[M]}),
\]
i.e. the Fr{\'{e}}chet variance explained by the first $M$ principal components is $V_{[M]} :=\operatorname{Var}_d(\bm{\zeta}^{[M]})$ and the fraction arising in the definition of $\hat M$ in  \eqref{eq:Mhat} is the proportion that $V_{[M]}$ explains with respect to the total Fr{\'{e}}chet variance in $\bm \zeta$. 
A similar result has been obtained for density functions in \citet{PetersenMueller:2016}. Using the metric $d$ as introduced in Definition~\ref{def:metric}, our results show that the variance decomposition is  valid in the more complex space $\mathcal{G}$, that combines warping functions $\gamma_i$ and registered functions $w_i$.

\section{Application to Earthquake simulations}
\label{sec:earthquake}

Using the proposed Bayes space transformation $\Psi_B$ and multivariate functional principal component approach, we apply our method to a subset of data from a seismological \emph{in silico} experiment based on the Mw 6.7 1994 Northridge earthquake, a blind thrust event that was felt over 200,000km$^2$. 
The induced ground shaking exceeded engineering building codes and resulted in sixty fatalities, >7,000 injured, 40,000 buildings damaged and 44 Billion \$ in economic losses. The hypocenter was located at about 19 km depth on a fault dipping southward at about 35$^\circ$ below the San Fernando Valley in the Los Angeles metropolitan area \citep{Hauksson1995}. The strongly pronounced topographic relief in the vicinity of the fault is associated with dextral transpression at the Pacific-North America plate boundary \citep{Montgomery1993}. The patterns of damage that occurred during the Northridge event showed irregular distributions. Generally, the region closest to the earthquake was shaken most severely. However, there were also isolated pockets of damage at distant locations.

Estimation of realistic ground motions for complex surface topography is a long-standing challenge in computational seismology \citep[e.g.][]{Chaljub2010}. 
The influence of topography on local site responses to seismicity are challenging to consider in current seismic hazard assessment typically based on empirical ground motion prediction equations, even though surface topography can have significant impact on seismic wave propagation and earthquake ground motions \citep{Bouchon1973,Boore1973}. For the data analysed here, physics-based earthquake scenarios were modelled using 3D unstructured meshes \citep{Rettenberger2016} constructed from geological constraints such as high-resolution topography data and the SCEC Community Fault Model combined with a 1D subsurface structure \citep{Wald1996}. Using SeisSol, an open source earthquake simulation package that couples 3D seismic wave propagation to the simulation of dynamic rupture propagation across earthquake fault zones \citep{Pelties2014, Heinecke2014, Uphoff2017}, multiple simulations varying the initial fault stress and strength conditions were performed. Time series of absolute ground velocity were recorded at a dense network of virtual seismometers distributed across Southern California. A more detailed description of the data and an analysis on the full set of original, unregistered, data can be found e.g.\ in \citet{BauerEtAl:2018}.

For our example, we focus on seismometer locations  with a maximum distance of 40~km from the hypocenter from two simulations which showed the strongest ground velocity movements on average. This results in a total of 1,558 observation units, each of which is recorded at 2~Hz over 30 seconds for a total of 61 timepoints. We pre-smoothed the ground velocity curves using a Tweedie distribution with log-link for the response and 40 cubic regression splines with penalized second derivative using R-package \texttt{mgcv} \citep{Wood:Pya:Saefken:2016} to model the evolution over time before registration. 
The corresponding smoothed curves representing $\log(1 + V_{i}(t))$ for ground velocity $V$ at locations $s_i$ and time $t$ are shown in Fig.~\ref{fig:appData}. The figure also shows the warping functions and the aligned functions after SRVF-based warping \citep{Tucker:2014}. This is the same warping approach as used in \citet{LeeJung:2017}, but other approaches are possible instead. The results of the warping are the inputs of the analysis. The warping functions $\gamma_i$ are transformed to $L^2(\mathcal{T})$ via the Bayes-space transform  $\Psi_B$. Together with the registered functions $w_i$, the transformed warping functions are then  fed into the MFPCA, which is calculated using the R-package \texttt{MFPCA} \citep{MFPCA}. We choose $M = 10$ principal components, for which the optimal weight is found to be  $\hat C = 1.03$. The weight of the aligned functions $w_i$ in the MFPCA thus approximately equals that of the transformed warping functions.

The first principal component, which explains $26.8\%$ of the (Fr{\'{e}}chet) variability in the data, is illustrated in Fig.~\ref{fig:appPC}.  Positive scores in this component are associated with a larger amplitude than the mean with five clearly defined peaks that occur earlier than in the mean. For negative scores one observes a lower amplitude and slightly delayed peaks. Visualizing phase and amplitude variation for this component separately (middle and right panels, respectively) shows that the phase variation in this first principal component is a fairly constant shift in time up to the last peak, whereas its main characteristic is driven by amplitude variation. 
In other words, for positive scores in the first principal component, ground movement is stronger and begins slightly earlier, whereas for negative scores, ground movement is somewhat delayed and markedly less severe. Therefore, one would  naturally expect rather positive scores for the first principal component for seismometers closer to the hypocenter and more negative values for those further away.

The second principal component (representing $15.6\%$ of Fr{\'{e}}chet variability) shows a different pattern. Here the main component is phase variation, whereas amplitude variation plays only a minor role affecting local maxima and minima after 20s, see Fig.~\ref{fig:appPC}, bottom row. For functions with positive scores, we observe that peaks occur earlier than in the mean function and that this temporal shift increases over time. Moreover, the amplitude tends to decrease faster than in the mean. Conversely, for functions with negative scores we observe five increasingly delayed peaks with more slowly decreased amplitude after 20s.

Fig.~\ref{fig:appScores} contains a scatter plot of the first two principal component scores colored by hypocentral distance (left panel) and the spatial distribution of the first two PC score vectors (middle \& right panel). 
As expected, seismometers which are closer to the epicenter are characterized by positive scores for the first principal component (earlier and larger movement), i.e.\ having scores in the right part of the first scatterplot, whereas seismometers with a higher hypocentral distance show the opposite behaviour.
We observe pronounced directivity effects caused by the unilateral earthquake rupture and dynamically enforced by the fault morphology. 
The dipping orientation of the fault adds further spatial complexity to the ground shaking. 
 The middle panel of Fig.~\ref{fig:appScores} shows larger and earlier ground motions in lighter colors. 
Besides source effects, topographic effects such as amplification, deamplification, scattering, and channeling of seismic waves strongly affect the seismic data analyzed. 
Concerning the second principal component, the image is more heterogeneous with some local spots having especially high values, i.e.\ here the ground movements arrive earlier but also decrease more rapidly. 

The seismological \emph{in silico} experiment presented in this section is exemplary for data in which both amplitude and phase variation are of interest. Using the Bayes space transformation for the warping functions, we preserve as much of the structure in the data as possible, as detailed in Section~\ref{sec:trafoFuns}. The bivariate functional principal components of the transformed warping functions and the registered functions are interpretable and provide interesting insights into the joint variation of amplitude and phase. In addition, the results found in Section~\ref{sec:Frechet} allow us to quantify the 
proportion of variability in the data explained by each principal component 
in terms of the Fr{\'{e}}chet variance using the metric defined in Definition~\ref{def:metric}. 
It is expected that more in-depth analyses on the complete data  will shed light on long-standing questions of source and path effects on ground motions.

\begin{figure}
\begin{center}
\includegraphics[width = \textwidth]{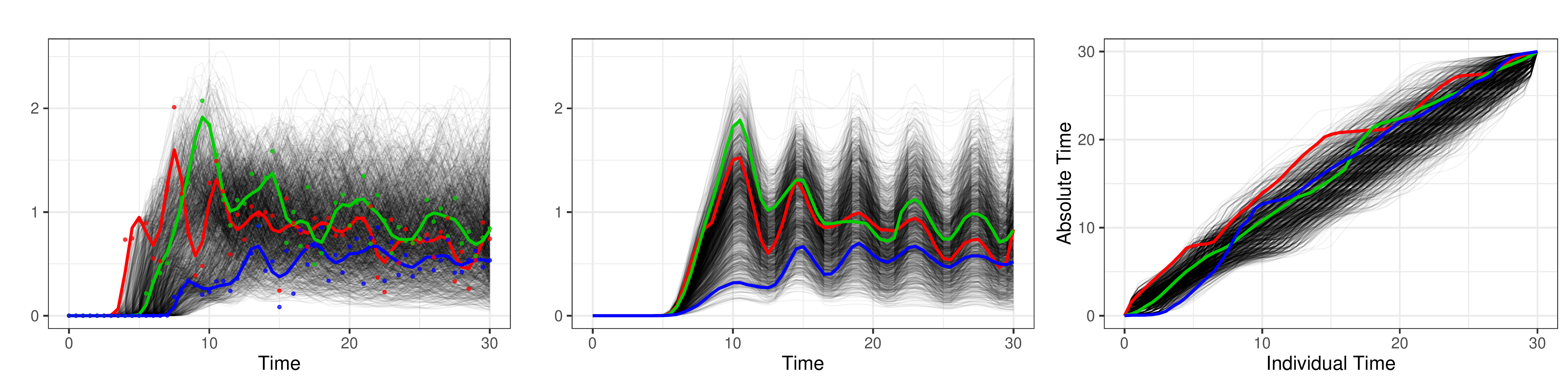}
\end{center}

\caption{Earthquake dataset representing 1,558 observations of simulated ground velocity over time (log-transformed). Left: smoothed curves. Middle: registered functions. Right: warping functions. The red (green/blue) curves have minimum (median/maximum) hypocentral distance in the dataset.}
\label{fig:appData}
\end{figure}

\begin{figure}
\begin{center}
\includegraphics[width = \textwidth]{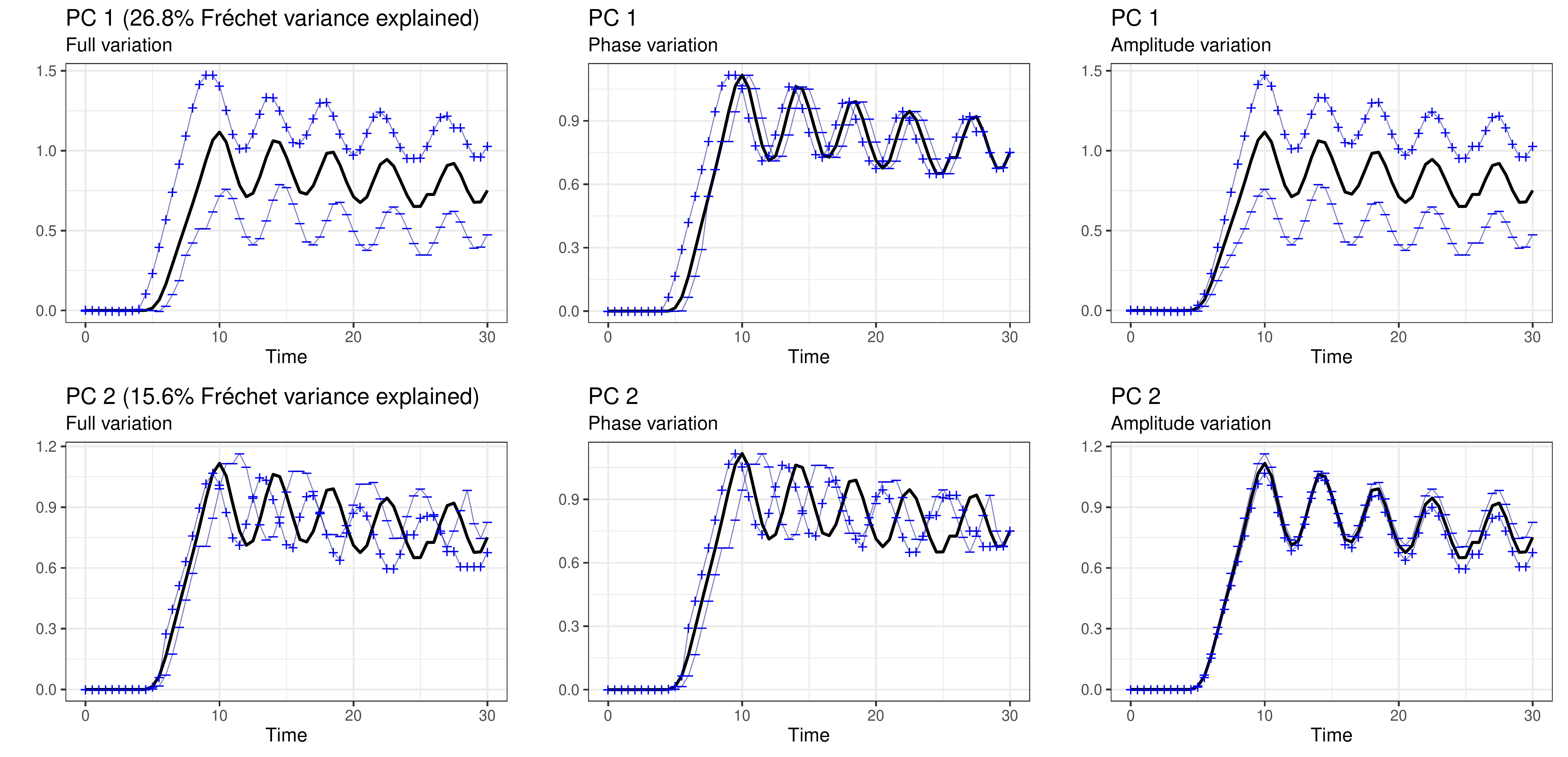}
\end{center}
\caption{Results for the first two principal components of the joint PCA for the earthquake dataset. The thick black curves represent the mean $\bar w \circ \Psi_B^{-1}(\bar v)$ and the blue lines correspond to  $(\bar w \pm \alpha_w \hat \nu_m^{1/2} \hat\phi_m\h{w} ) \circ \Psi_B ^{-1}(\bar v \pm  \alpha_v \hat \nu_m^{1/2} \hat\phi_m\h{v})$. 
First / second row: Effect of the first ($m = 1$) / second ($m = 2$) principal component.
  Left: Joint effect ($\alpha_w = 1, \alpha_v = 1$). Middle: Phase effect ($\alpha_w = 0, \alpha_v = 1$). Right: Amplitude effect ($\alpha_w = 1, \alpha_v = 0$).
Note that these visualizations show the absolute time of the aligned functions on the horizontal axis.
}
\label{fig:appPC}
\end{figure}

\begin{figure}
\begin{center}
\includegraphics[width = \textwidth]{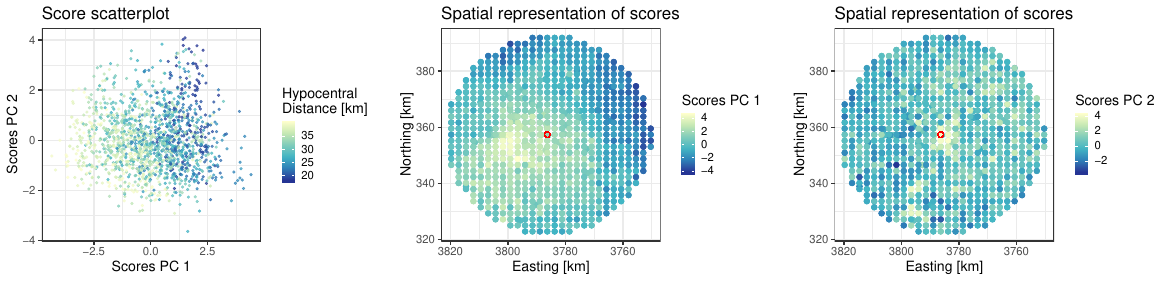}
\end{center}
\caption{Scores for the first two principal components of the combined PCA. Left: Scatterplot of the scores $\rho_{i1}, \rho_{i2}$ depending on the hypocentral distance. Middle / Right: Spatial distribution of the scores for the first / second principal component. 
Spatial locations of the seismometers are given in UTM coordinates. The red circle corresponds to the seismometer with minimal hypocentral distance (red curve in Fig.~\ref{fig:appData}).}
\label{fig:appScores}
\end{figure} 

\section*{Supplementary Material}
The supplementary files contain \texttt{R} code and data to fully reproduce the toy example in Section~\ref{sec:trafoFuns} and the application in Section~\ref{sec:earthquake}.

\section*{Acknowledgements}
We would like to thank Almond St{\"o}cker for drawing our attention to the theory of Bayes spaces for probability densities and David R{\"u}gamer for his valuable comments on the text. Regarding the seismological application, we are grateful to Helmut K{\"u}chenhoff and Alexander Bauer for their collaboration.


\bibliographystyle{apalike}
\bibliography{Happ_etal}

\begin{figure}
\begin{center}
\includegraphics[width = \textwidth]{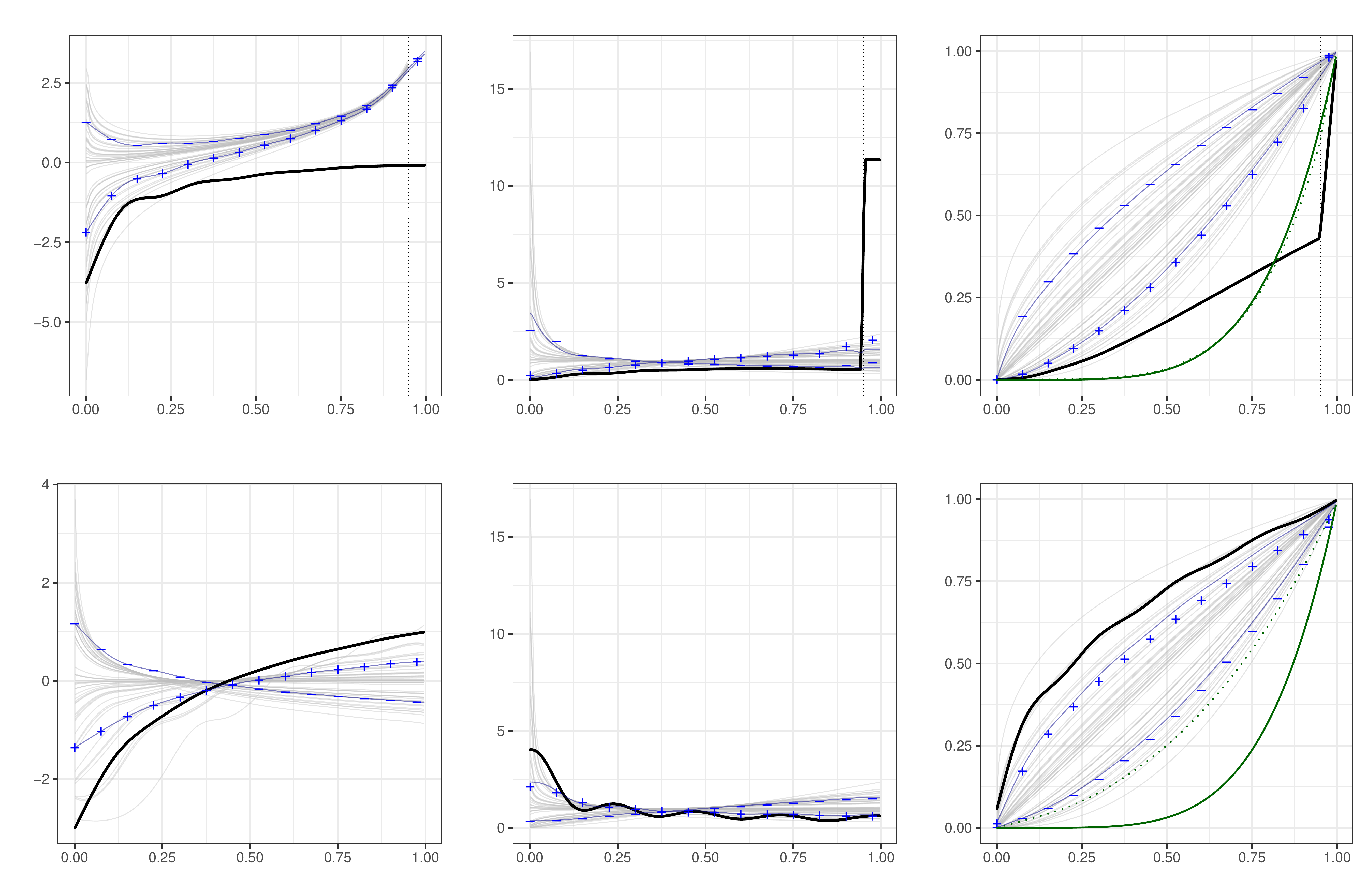}
\end{center}
\caption{
PCA of warping functions based on $\psi_H$ (top row) or $\psi_Q$ (bottom row) for the same warping functions as in Fig.~\ref{fig:exampleSRVF} and different stages of transformation. The grey functions represent the (transformed) data, the black function the first principal component $\hat \phi_1$. Blue lines visualize the effect of the first principal component as perturbation from the mean ($\hat \mu \pm \hat \lambda_1^{1/2} \hat \phi_1$).
Left: $L^2(\mathcal{T})$. Center: Density space $B^2(\mathcal{T})$. Right: Space of warping functions $\Gamma(\mathcal{T})$.
The dotted vertical line in the top row corresponds to the threshold parameter $\delta$ for the log-hazard transformation $\psi_H$, here chosen as $\delta = 0.05$.
The green curves corresponds to a new observation (solid) and its reconstruction (dotted) based on the first principal component.
}
\label{fig:exampleLH}
\end{figure}

\appendix

\begin{proof}[Proof of equation \eqref{eq:imdomPsiS}]
The first equation is clear by the definition of the mappings $\tilde \psi_{S,\mu},~\tilde \psi_{S,\mu}^{-1}$ and the spaces $S_+^\infty(\mathcal{T}),~S^\infty(\mathcal{T})$.

For the second equation, choose for example $\mu = q_0$ the SRVF associated with identity warping. 
By the Cauchy-Schwarz inequality, $\left| \scal{q}{\mu} \right | \leq \norm{q} \norm{\mu} = \eta
$ for all $q,\mu \in S_+^\infty(\mathcal{T})$.
Using that $q(t) \geq 0$ for all $t\in \mathcal{T}$ further yields
$
\scal{q}{\mu}= \int_a^b q(t) \mathrm{d} t \geq 0
$.
Therefore, 
$
\frac{\scal{q}{\mu}}{\eta} \in [0,1]$ and $\theta \in \left[0, \frac{\pi}{2} \right]
$.

Let now $v = \tilde \psi_{S,\mu}(q)$. It is easy to show that $\normsq{v} = \theta^2 $ \citep{SrivastavaBook}, which implies $\norm{v} \leq \frac{\pi}{2} $. The boundary cases are given by $\norm{v} \to 0$ for $q \to \mu$, hence $\scal{q}{\mu} \to \eta$ and  $\norm{v} \to \frac{\pi}{2}$ for $\scal{q}{\mu} =0$, meaning that $q \in T_\mu(\mathcal{T})$ already. Note that the latter case is only achieved by degenerate functions $q$ that correspond to stepwise constant warping functions $\gamma$.

Using $q(t) \geq 0$ for all $t \in \mathcal{T}$ further yields
\[
v(t) 
= \frac{\theta}{\eta^{1/2}\sin(\theta)} (q(t)- \cos(\theta) \mu(t))
\geq \frac{q(t) - \cos(\theta)}{\eta^{1/2}}
 \geq -\eta^{-1/2}.
\]
This shows
\[
\operatorname{im}(\tilde \psi_{S,\mu} )
\subseteq \left \{v \in T_\mu(\mathcal{T}) \colon \norm{v} \leq \frac{\pi}{2}, v(t) \geq -\eta^{-1/2} ~ \forall~ t \in \mathcal{T}\right\}
\subsetneq T_\mu(\mathcal{T}).
\]
Since $\norm{v} < \pi$ for all $v$ in $\operatorname{im}(\tilde \psi_{S,\mu} )$, the restriction of $\tilde \psi_{S,\mu}^{-1}$ to $\operatorname{im}(\tilde \psi_{S,\mu} )$ is a bijection \citep{SrivastavaBook}. For alternative choices of $\mu$, analogous results can be obtained.

\end{proof}

\end{document}